\documentclass[9pt, conference]{IEEEtran}
\IEEEoverridecommandlockouts
\usepackage{cite}
\usepackage{amsmath,amssymb,amsfonts}
\usepackage{algorithmic}
\usepackage{graphicx}
\usepackage{textcomp}
\usepackage{xcolor}

\usepackage{multirow}

\newtheorem{theorem}{Theorem}

\newtheorem{proof}{Proof}

\newcommand{\qed}{\hfill\ensuremath{\blacksquare}}

\def \VersionWithComments {}

\ifdefined \VersionWithComments
\usepackage{marginnote}
\newcommand{\marginX}{\marginnote{\huge{\quad\quad\textbf{!}\quad\quad}}}
\newcommand{\cyr}[1]{\mbox{}{\color{green!50!black}\marginX{}\textbf{[Yean-Ru}: #1]}}
\newcommand{\lsw}[1]{\mbox{}{\color{orange}\marginX{}\textbf{[Shang-Wei}: #1]}}
\newcommand{\instructions}[1]{{\color{red}\marginX{}\textbf{[Instructions: ``#1'']}}}
\newcommand{\reviewer}[2]{\mbox{}{\color{red}\marginX{}\textbf{[Reviewer #1}: ``#2'']}}
\newcommand{\todo}[1]{\mbox{}{\color{blue}{\marginX{}\textbf{TODO}\ifx#1\\\else:\ \fi #1}}} 
\else
\newcommand{\instructions}[1]{}
\newcommand{\cyr}[1]{}
\newcommand{\lsw}[1]{}
\newcommand{\reviewer}[2]{}
\newcommand{\todo}[1]{}
\fi

\def\BibTeX{{\rm B\kern-.05em{\sc i\kern-.025em b}\kern-.08em
    T\kern-.1667em\lower.7ex\hbox{E}\kern-.125emX}}
\begin{document}

\title{A Quantum SMT Solver for Bit-Vector Theory}

 \author{
 
 \IEEEauthorblockN{Shang-Wei Lin}
 \IEEEauthorblockA{\textit{Nanyang Technological University, Singapore} \\
 shang-wei.lin@ntu.edu.sg}
 \and
 \IEEEauthorblockN{Si-Han Chen, Tzu-Fan Wang and Yean-Ru Chen}
 \IEEEauthorblockA{\textit{National Cheng Kung University, Taiwan} \\
 chenyr@mail.ncku.edu.tw}
}

\maketitle
\thispagestyle{plain}
\pagestyle{plain}

\begin{abstract}

Given a formula $F$ of satisfiability modulo theory (SMT), the classical SMT solver tries to (1) abstract $F$ as a Boolean formula $F_B$, (2) find a Boolean solution to $F_B$, and (3) check whether the Boolean solution is consistent with the theory. Steps~{(2)} and (3) may need to be performed back and forth until a consistent solution is found. In this work, we develop a quantum SMT solver for the bit-vector theory. With the characteristic of superposition in quantum system, our solver is able to consider all the inputs simultaneously and check their consistency between Boolean and the theory domains in one shot.
\end{abstract}


\section{Introduction} \label{sec:Introduction}

Satisfiability modulo theory (SMT) is the problem of determining the satisfiability of a first-order formula with respect to a decidable first-order theory. Compared to Boolean satisfiability problem (SAT), the atoms in SMT can be in different forms depending on the theories they base on. With this richer expressiveness, SMT is widely used in plenty of applications, particularly in formal verification, ranging from recent topics like neural networks \cite{KB17} \cite{AW21} and smart contracts \cite{AR18} \cite{AG20} to hardware circuit designs \cite{FS17} \cite{AF18}. Consider the following formula $F$, based on the theory of fixed-width bit-vector, consisting of three integer variables $a$, $b$ and $c$, each of which is $2$-bit long.
\[
F: [(a > b) \vee (a < b) \vee (a = b)] \wedge [\neg(a > b) \vee \neg(a < b) \vee \neg(a = b)]
\]

Formula $F$ is {\em satisfiable} as there exists an assignment $a = 2$ and $b = 2$, with binary representation $10$ and $10$, respectively, making $F$ evaluate to True. Such an assignment is called a {\em solution}. If there does not exist any solution, the formula is called {\em unsatisfiable}.

To solve a SMT problem, most of state-of-the-art SMT solvers adopt the classical {\em lazy approach}~\cite{S07}, which consists of three steps: (1) abstract the original formula as a Boolean formula, i.e., make an abstraction from the original theory domain to the Boolean domain, (2) find a solution to the abstract Boolean formula, and (3) check the consistency between Boolean and the original theory domains. If the solution found in step~{($2$)} is not consistent in both domains, it will be abandoned, and step~{(2)} continues until a consistent solution is found, or the formula is concluded unsatisfiable.


For instance, if we use Boolean variables $x$, $y$, and $z$ to denote the predicate $(a > b)$, $(a < b)$, and $(a = b)$, respectively, formula $F$ can be abstracted as the following Boolean formula $F_B$:
\[
F_B: (x \vee y \vee z) \wedge (\neg x \vee \neg y \vee \neg z)
\]
An SAT solver can be used to find a solution to the Boolean formula $F_B$, and a bit-vector theory solver can help to check the consistency. Fig.~\ref{fig:SMT_lazy_example} illustrates how the lazy approach solves formula $F$. In the first three iterations, SAT solver finds an assignment but immediately rejected by the theory solver. For example, in iteration~$1$, SAT solver finds a solution $x=0 \wedge y=1 \wedge z=1$ to formula $F_B$, which is rejected by the theory solver because $(a < b)$ and $(a = b)$ cannot hold simultaneously. This back and forth process continues two more iterations, and a consistent solution is found in iteration~$4$.

\begin{figure}[tb]
    \centering
    \includegraphics[width=0.98\linewidth]{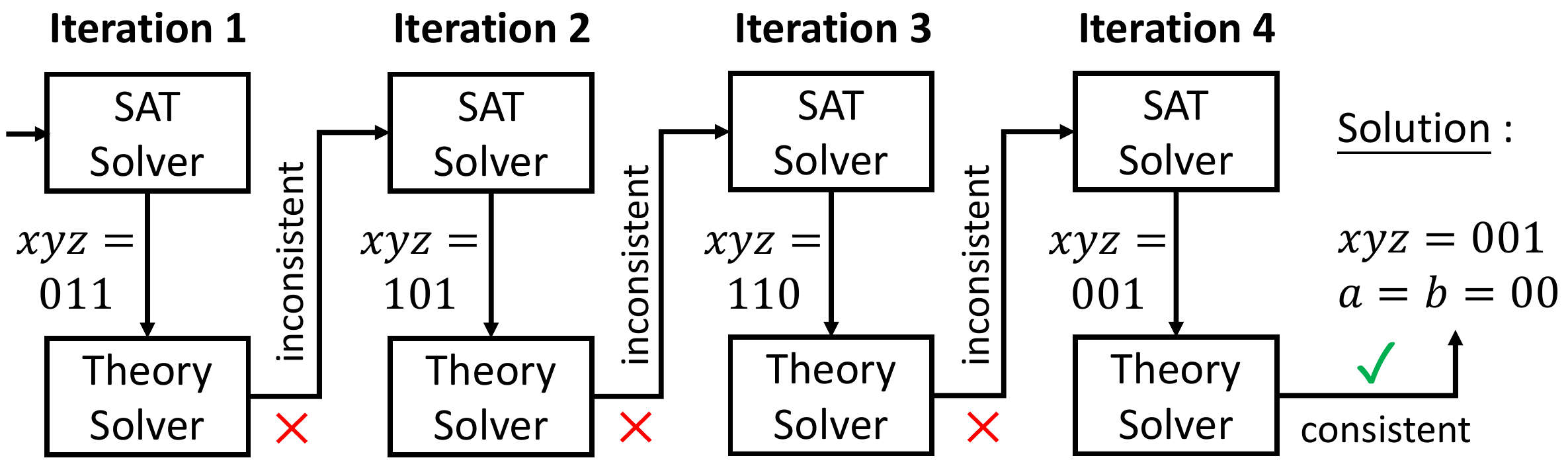}
    \caption{SMT solving by the classical lazy approach}
    \label{fig:SMT_lazy_example}
\end{figure}


One can observe that each iteration of this back and forth process only checks the consistency of one Boolean solution (against the theory) at a time. If the search space is huge, this approach becomes unscalable. Indeed, there is an {\em online} version of the lazy approach, which adopts heuristics, making it more efficient. However, the online version does not reduce the theoretical time complexity.


In recent years, quantum technology is widely used in many applications to solve traditionally difficult problems by taking advantage of the nature of superposition and entanglement in quantum systems. Grover’s algorithm~\cite{G96}, one of the famous quantum algorithms, helps to search target objects among a huge search space. There are two essential components in Grover's algorithm: (1) an {\em oracle}, and (2) the {\em diffuser}. In a nutshell, the oracle answers the ``yes/no'' question about whether an object in the search space is our target, while the diffuser tries to increase or maximize the probability of the target's being measured. The details of Grover's algorithm is briefed in Section~\ref{sec:Preliminary}.

To use Grover's algorithm for a search problem, the key is to provide the oracle. As long as the oracle can correctly identify the targets among the search space, the diffuser, which is standard and independent from the search problem, can help to ``extract'' the targets. In this work, we develop a quantum SMT solver based on Grover's algorithm. More specifically, given a SMT formula, we develop a methodology to generate the oracle required by Grover's algorithm to search the solutions to the SMT formula.

Let us take formula $F$ as an example again. Assume the two variables $a$ and $b$, each of which is $2$-bit long, are represented by binary strings $a_1 a_2$ and $b_1 b_2$, respectively. Together with the abstract Boolean formula $F_B$, we can use seven bits to represent each element $|v\rangle$ in the search space, as a column vector $|x,y,z,a_1,a_2,b_1,b_2\rangle$. There are requirements for $|v\rangle$ to be a solution of formula $F$. Firstly, in the Boolean domain, $x$, $y$, $z$  have to satisfy the following condition:
\vspace{-1mm}
\begin{equation} \label{eq:Boolean}
F_B(x,y,z) = (x \vee y \vee z) \wedge (\neg x \vee \neg y \vee \neg z) = 1
\end{equation}

Secondly, in the bit-vector theory domain, $a_1$, $a_2$, $b_1$, $b_2$ need to be consistent with their Boolean abstraction $x$, $y$, $z$. Thus, they have to satisfy the following three conditions:
\vspace{-1mm}
\begin{eqnarray}
(x = 1 \iff a_1 a_2 > b_1 b_2) \vee  (x = 0 \iff a_1 a_2 \not> b_1 b_2) \\
(y = 1 \iff a_1 a_2 < b_1 b_2) \vee  (y = 0 \iff a_1 a_2 \not< b_1 b_2) \\
(z = 1 \iff a_1 a_2 = b_1 b_2) \vee  (z = 0 \iff a_1 a_2 \neq b_1 b_2)
\end{eqnarray}

Here comes the interesting and critical part. If we can construct an oracle $\Psi$, which can help us to take care of the four conditions simultaneously, then the diffuser can proceed to extract the solution for us. That is, our oracle $\Psi$ does the following:

\begin{displaymath}
\Psi(|v\rangle) = \left\{ 
\begin{matrix}
 -1 \cdot |v\rangle & \mbox{if} & (1) \wedge (2) \wedge (3) \wedge (4) \\
 |v\rangle & \mbox{otherwise} &
\end{matrix}
\right.
\end{displaymath}

If an input $|v\rangle$ is a solution to $F$, our oracle adds a ``$-1$'' phase to it; otherwise, $|v\rangle$ is not changed. Since the input $|v\rangle$ can be placed in superposition, representing every possible input, our oracle is able to identify all the solutions to $F$ in one shot. Then, the diffuser recognizes those inputs with a ``$-1$'' phase and tries to increase/maximize their probability to be measured.
Fig.~\ref{fig:Oracle_High-level}~{(b)} shows all the $16$ solutions to formula $F$, and our oracle is able to identify all of them among the $128$ inputs in one shot.


Fig.~\ref{fig:Oracle_High-level}~{(a)} shows the structure of our oracle design and summarizes our contributions. There are four components in our oracle:
\begin{enumerate}
\item {\bf SAT Circuit} (c.f.~Section~\ref{subsec:SAT-Circuit-Degisn}) determines the solutions for the Boolean domain, i.e., Equation~(\ref{eq:Boolean}) in our example.

\item {\bf Theory Circuit} (c.f.~Section~\ref{subsec:Theory-Circuit}) determines the truth value of predicates (e.g., $a < b$, $a > b$, $a = b$ in our example) in the theory domain with respect to different inputs.

\item {\bf Consistency Extractor} (c.f.~Section~\ref{subsec:Consistency_Extractor}) identifies those solutions that are consistent in both Boolean and the theory domains, i.e., $(1) \wedge (2) \wedge (3) \wedge (4)$ in our example.

\item {\bf Solution Inverser} (c.f.~Section~\ref{subsec:solution_inverter}) inverses the solutions to the SMT formula by adding a ``-1'' phase to them for the diffuser's further processing.
\end{enumerate}

\begin{figure}[tb]
\begin{minipage}{0.6\linewidth}
\centering
\includegraphics[width=0.9\linewidth]{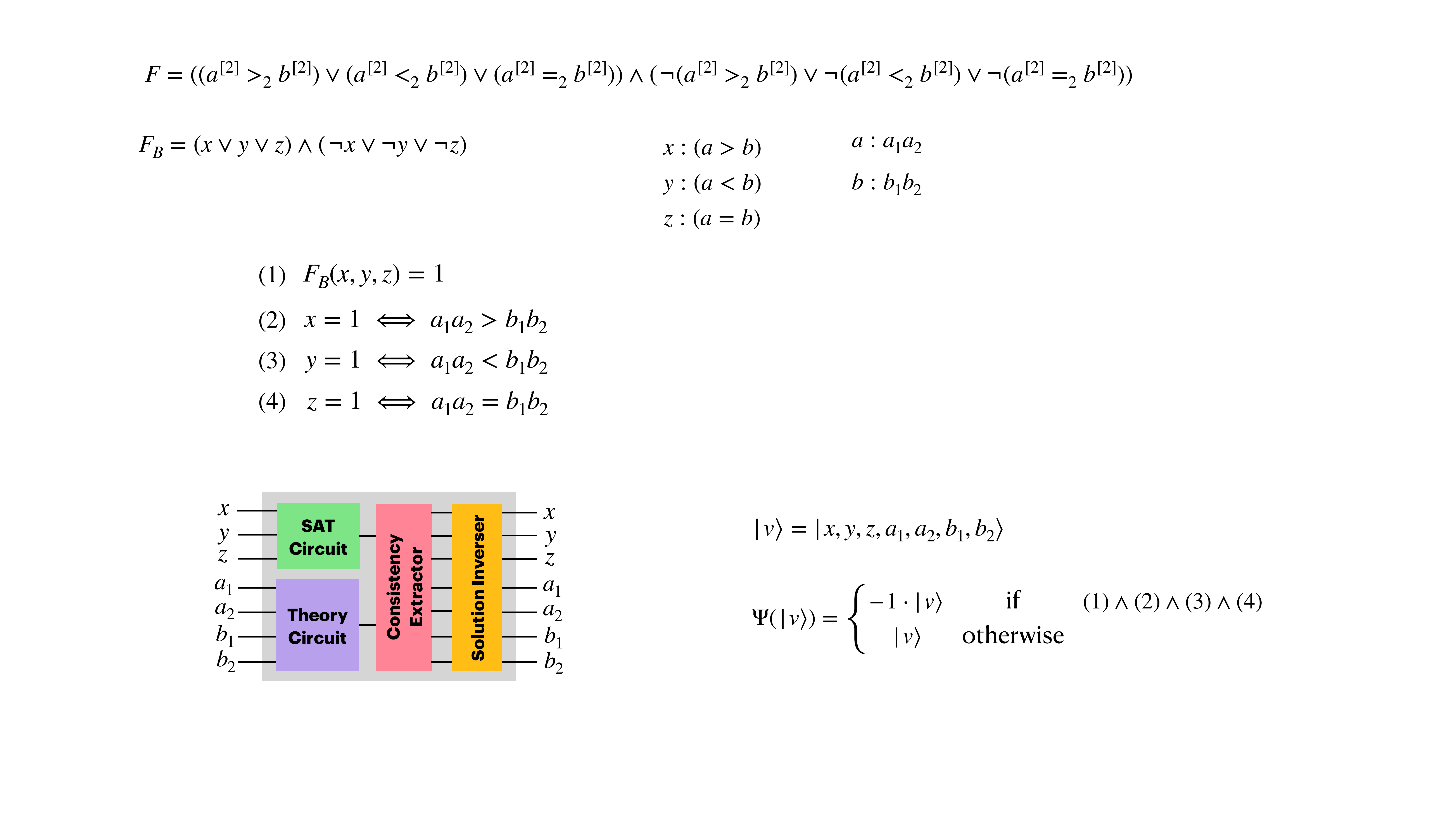}
\end{minipage}
~
\begin{minipage}{0.35\linewidth}
\centering
\includegraphics[width=0.95\linewidth]{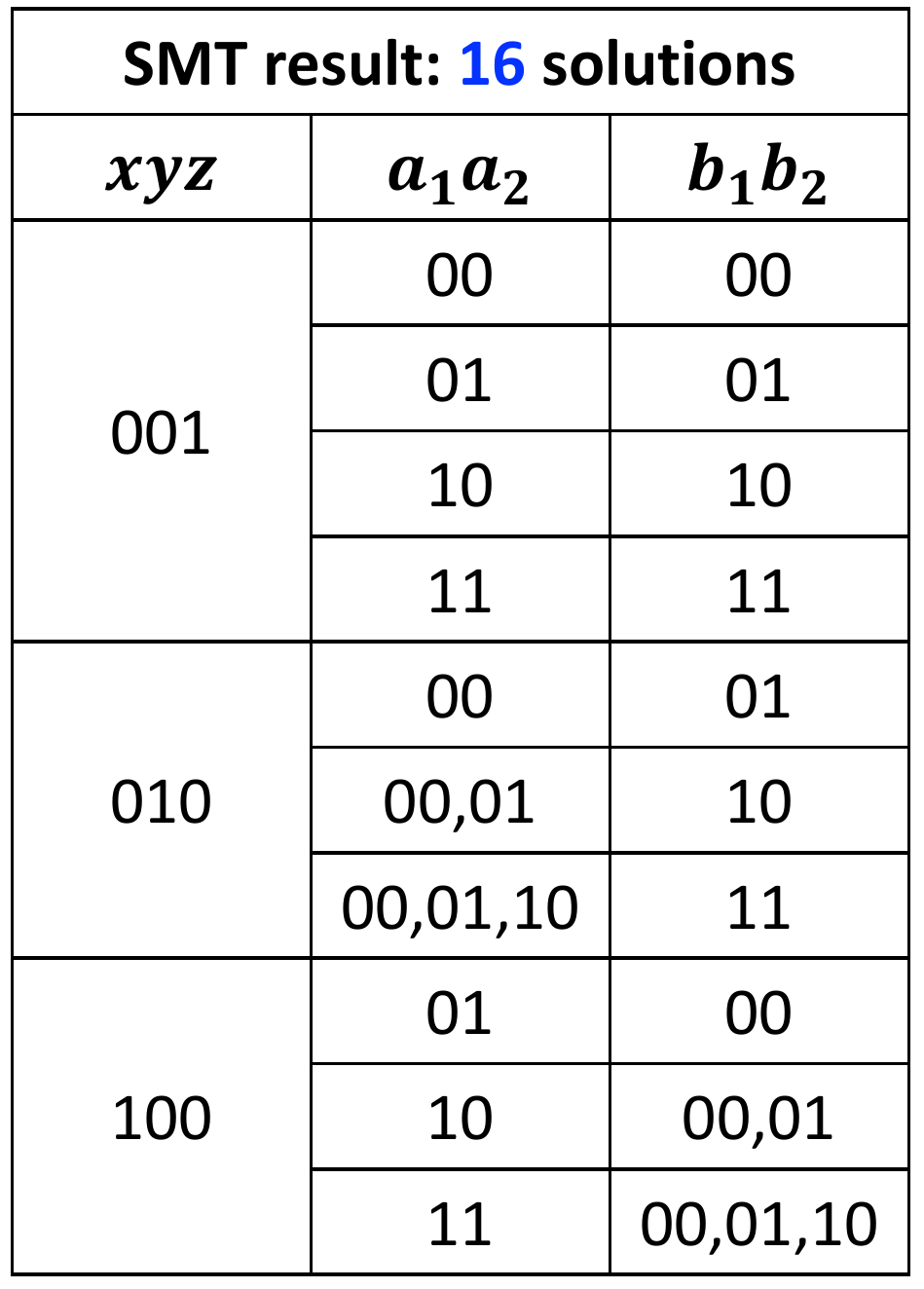}
\end{minipage} \\[1mm]
\begin{minipage}{0.6\linewidth}
\centering
\footnotesize
(a) Our oracle $\Psi$
\end{minipage}
~
\begin{minipage}{0.35\linewidth}
\centering
\footnotesize
(b) Solutions
\end{minipage}
\caption{Structure of our oracle $\Psi$}
\label{fig:Oracle_High-level}
\end{figure}


We have also proved the correctness of each component design. The rest of this paper is organized as follows. 
Section~\ref{sec:Preliminary} reviews preliminary concepts about quantum computing.
The proposed quantum SMT solver based on Grover's algorithm for formulas of the bit-vector theory is introduced in Section~\ref{sec:Methodology}. 
Evaluations of our approach are given in Section~\ref{sec:Evaluation}. Related works are discussed in Section~\ref{sec:RelatedWorks}.
The conclusion and our future works are discussed in Section~\ref{sec:Conclusion}.

\section{Preliminary} \label{sec:Preliminary}

We assume that our readers have basic concepts in quantum computing, e.g., the {\em tensor product} operation, {\em inner product} operation, {\em outer product} operation, and primitive quantum gates such as $\mathtt{X}$, $\mathtt{Z}$, $\mathtt{H}$, $\mathtt{CCNOT}$, etc. We use the {\em ket} notation $|\cdot \rangle$ to denote the (column) vector representing the state of a quantum system, and the {\em bra} notation $\langle \cdot |$ to denote its conjugate transpose. Given two vectors $|v_1\rangle$ and $|v_2\rangle$, we use $\langle v_1 | v_2 \rangle$ to denote their inner product, $|v_1\rangle \langle v_2|$ for their outer product, and $|v_1\rangle \otimes |v_2\rangle$ for their tensor product. For simplicity, we may write $|v_1\rangle \otimes |v_2\rangle$ as $|v_1,v_2\rangle$, or even $|v_1 v_2\rangle$.

\subsection{Grover’s algorithm}
Grover’s algorithm~\cite{G96} is one of the most famous quantum algorithms. It is used to solve the searching problem for finding one target element in a disordered database with $N$ elements. 
Due to the characteristic of parallel computation in quantum systems, Grover's algorithm takes $O(\sqrt{N})$ operations to find the target element, which is a quadratic speed up compared with classical methods requiring $O(N)$ operations.
Grover’s algorithm is widely used in many applications, such as cryptography \cite{GL16}, pattern matching \cite{TN22}, etc.

\begin{figure}[tb]
    \centering
    \includegraphics[width=0.92\linewidth]{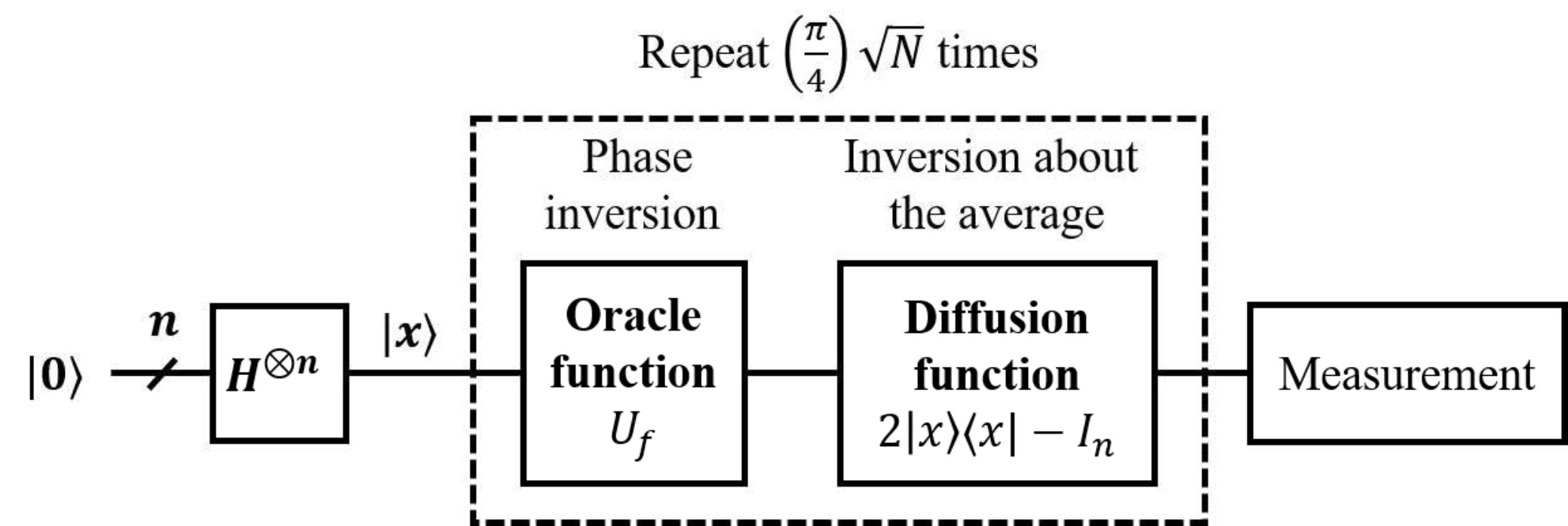}
    \caption{Grover's algorithm}
    \label{fig:grover_implement}
\end{figure}

The overall structure of Grover’s algorithm is shown in Fig.~\ref{fig:grover_implement}. The two main operations of it are {\em phase inversion} and {\em inversion about the average}, which are handled by the oracle and diffuser, respectively. 
Initially, the input will be placed in superposition ($|x\rangle$) to evaluate all elements in the database at once. 
Next, the oracle function $U_f$ considers all the possible inputs and marks the target element by applying phase inversion, i.e., $U_f |x\rangle = (-1)^{f(x)} |x\rangle$, in which $f(x) = 1$ for the target element and $f(x) = 0$ for the others.
The oracle is problem-dependent, while the diffuser is not. Thus, designing the correct oracle is the key to apply Grover’s algorithm. 
After the target element being marked, the diffuser applies the {\em inversion about the mean} operation, to amplify the probability of the target element, so that one can obtain the result by measurement.
In order to achieve the optimal probability for the target element to be measured, the two operations (called a Grover iteration) need to repeat $(\pi / 4) \sqrt{N}$ iterations .

Grover’s algorithm can also be deployed for problems with multiple target elements. 
In such a case, the number of required iterations becomes $(\pi / 4) \sqrt{N/M}$, where $M$ is the number of target elements. 
In addition, one can measure the correct result only when the number of the target elements is less than that of the nontarget elements (i.e. $M/N < 50\%$) due to the natural property of the algorithm~\cite{Y08}.
Since the number of target elements is usually unknown before the searching, there are several ways to solve this issue.
The most common one is to apply quantum counting~\cite{BH98} to obtain the (approximate) number of target elements before using Grover’s algorithm. After knowing the ratio of $M/N$, if it is more than $50\%$, one can double the search space by adding $N$ nontarget elements, so the target elements will always be less than $50\%$ of the search space. Another way is to apply the advanced algorithm proposed by Boyer et al.~\cite{BB98}.

\subsection{SAT, SMT, and Theory of Fixed-Width Bit Vector}

A Boolean formula is constructed by logical operators such as negation ($\neg$), conjunction ($\wedge$) and disjunction ($\vee$) among Boolean variables. Given a Boolean formula, the Boolean satisfiability (SAT) problem asks if there exists an assignment of the Boolean variables that satisfies the formula. 

Satisfiability modulo theory (SMT), an extension of SAT, is the problem of determining whether a mathematical formula is satisfiable. A mathematical formula is composed of atoms (predicates) connected by Boolean logical operators (e.g., formula $F$ in Section~\ref{sec:Introduction}). The type of atoms depends on the theory used in the mathematical formula. 

\begin{figure}[tb]
\centering
\[
\begin{matrix}
    E & \simeq & \Tilde{a} \mid \mathcal{C} \mid \; E_1 \odot E_2 & \mbox{where} & \odot \mbox{ is an arithmetic operator} \\
    atom & \simeq & E_1 \; \rhd \; E_2 & \mbox{where} & \rhd \in \{<, >, =, \geq, \leq, \neq \}
\end{matrix}
\]
\caption{Syntax of $\mathcal{BV}$ atoms}
\label{fig:BV-Syntax}
\end{figure}

In this work, we focus on the theory of quantifier free fixed-width bit-vectors (c.f.~Section~4.2.5 in \cite{S07}), denoted by $\mathcal{BV}$ for short. The syntax of $\mathcal{BV}$ atoms is given in Fig.~\ref{fig:BV-Syntax}. An expression $E$ could be a variable $\Tilde{a}$ with $n$ bits, a constant $\mathcal{C}$, or the result of an arithmetic operation between two expressions $E_1$ and $E_2$. Notice that once the data width ($n$ bits) is determined, it is fixed for all expressions. An atom is composed of two expressions and one comparison operator $\rhd$ in between, where $\rhd \in \{<, >, =, \geq, \leq, \neq \}$. The arithmetic operator includes word-concatenation, modulo-sum, modulo-multiplication, bitwise operation, shift operation, etc.



\section{Methodology} \label{sec:Methodology}

\begin{figure}[tb]
    \centering
    \includegraphics[width=\linewidth]{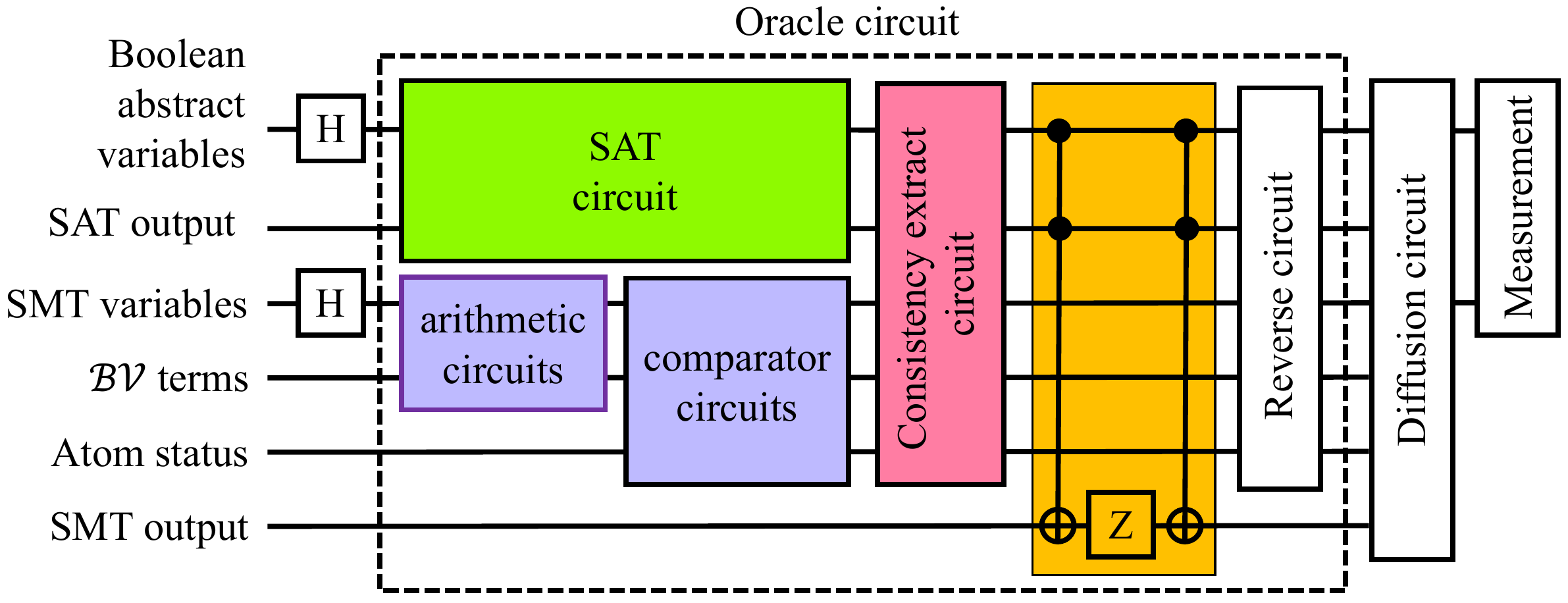}
    \caption{Quantum circuit for our $\mathcal{BV}$ SMT solver}
    \label{fig:SMT_circuit}
\end{figure}

In this section, we introduce our quantum SMT solver for the $\mathcal{BV}$ theory based on Grover's algorithm. Fig.~\ref{fig:SMT_circuit} shows its block diagram, echoing the overall structure in Fig~\ref{fig:Oracle_High-level}~{(a)}. As mentioned in Section~\ref{sec:Introduction}, the oracle is the key component for marking the solutions so that the difusser knows which elements to increase/maximize their probability to be measured. Since the diffuser is standard and independent from the search problem, we will not discuss it further. Instead, in the following, we put emphasise on how to construct the oracle.

\subsection{SAT Circuit} \label{subsec:SAT-Circuit-Degisn}

Fernandes et al.~\cite{FS19} showed an example of constructing the oracle circuit for a $3$-SAT problem and solved the problem based on Grover’s algorithm. However, they did not provide a systematic way to construct the oracle circuit. In this section, we propose a method to constructively generate the oracle circuit for an arbitrary $3$-SAT formula and prove its correctness. Consider the following grammar for $3$-SAT formulas in conjunctive normal form (CNF):
\[
\begin{matrix}
    F & \simeq & C \mid F_1 \wedge F_2 \\
    C & \simeq & l_1 \vee l_2 \vee l_3 \\
    l & \simeq & v \mid \neg v \mid 0
\end{matrix}
\]

Fig.~\ref{fig:SAT_clause_circuit} shows how to construct a quantum circuit for a clause $C: l_1 \vee l_2 \vee l_3$. Notice that the $Q$ gate depends on each literal $l_i$. If $l_i$ is a negative literal $\neg v_i$, $Q$ would be the $\mathtt{I}$ gate, i.e., the identity gate; otherwise, $Q$ would be the $\mathtt{X}$ gate, also called the $\mathtt{NOT}$ gate. The qubit $q_a(C)$ is an ancilla bit\footnote{The notation $q_a(C)$ is not a function application. It just represents that $q_a$ is the ancilla bit of clause $C$. Similarly, $q_o(C)$ represents the output bit of clause $C$, and $q_o(F)$ represents the output bit of formula $F$.} for internal computation, and $q_o(C)$ is the output bit for the truth value of $C$. Theorem~\ref{thm:SAT-Clause} proves that the quantum circuit construction for a clause $C$ is correct.

\begin{theorem} \label{thm:SAT-Clause}
{\bf [Clause Correctness]} (1) $q_o'(C) = 1$ $\iff$ clause $C$ is true, (2) $q_a'(C) = 0$, and (3) $v_i' = v_i$ for all $i \in \{1,2,3\}$.
\end{theorem}

\begin{proof}
Given a clause $C: l_1 \vee l_2 \vee l_3$, since if $l_i$ is a negative literal $\neg v_i$, $Q$ would be the $\mathtt{I}$ gate; otherwise, $Q$ would be the $\mathtt{X}$ gate, we have $Q(v_i) = \neg l_i$, as the red notations in Fig.~\ref{fig:SAT_clause_circuit}. Let $a$ be the value of $q_a(C)$ after apply the first $\mathtt{CCNOT}$ gate on it.

To prove condition~${(1)}$: $q_o' (C) = 0 \Leftrightarrow (\neg l_3 = 1) \wedge (a = 1)$ $\iff$ $(\neg l_3 = 1) \wedge (\neg l_1 = 1) \wedge (\neg l_2 = 1)$. If we apply negation on both side, we have
$\neg (q_o' (C) = 0) \Leftrightarrow \neg ((\neg l_1 = 1) \wedge (\neg l_2 = 1) \wedge (\neg l_3 = 1))$. Thus,
$q_o' (C) = 1 \Leftrightarrow \neg (\neg l_1 = 1) \vee \neg (\neg l_2 = 1) \vee \neg (\neg l_3 = 1)$. That is, 
$q_o' (C) = 1 \Leftrightarrow (l_1 = 1) \vee (l_2 = 1) \vee (l_3 = 1)$.

To prove condition~${(2)}$: $q_a' (C) = (\neg l_1 \wedge \neg l_2) \oplus a$. Since $a$ is obtained by applying $\mathtt{CCNOT}$ gate on $q_a(C)$, with $\neg l_1$ and $\neg l_2$ as the control bits, we have
$q_a' (C) = (\neg l_1 \wedge \neg l_2) \oplus ((\neg l_1 \wedge \neg l_2) \oplus 0)$
$= (\neg l_1 \wedge \neg l_2) \oplus (\neg l_1 \wedge \neg l_2) = 0$.

To prove condition~${(3)}$: $v_i' = Q(\neg l_i) = Q(Q(v_i)) = v_i$, because $QQ = I$, as $Q$ is either $\mathtt{X}$ or $\mathtt{I}$ gate. \qed
\end{proof}

\begin{figure}[tb]
\centering
\begin{minipage}{0.56 \linewidth}
\centering
\includegraphics[width=0.9\linewidth]{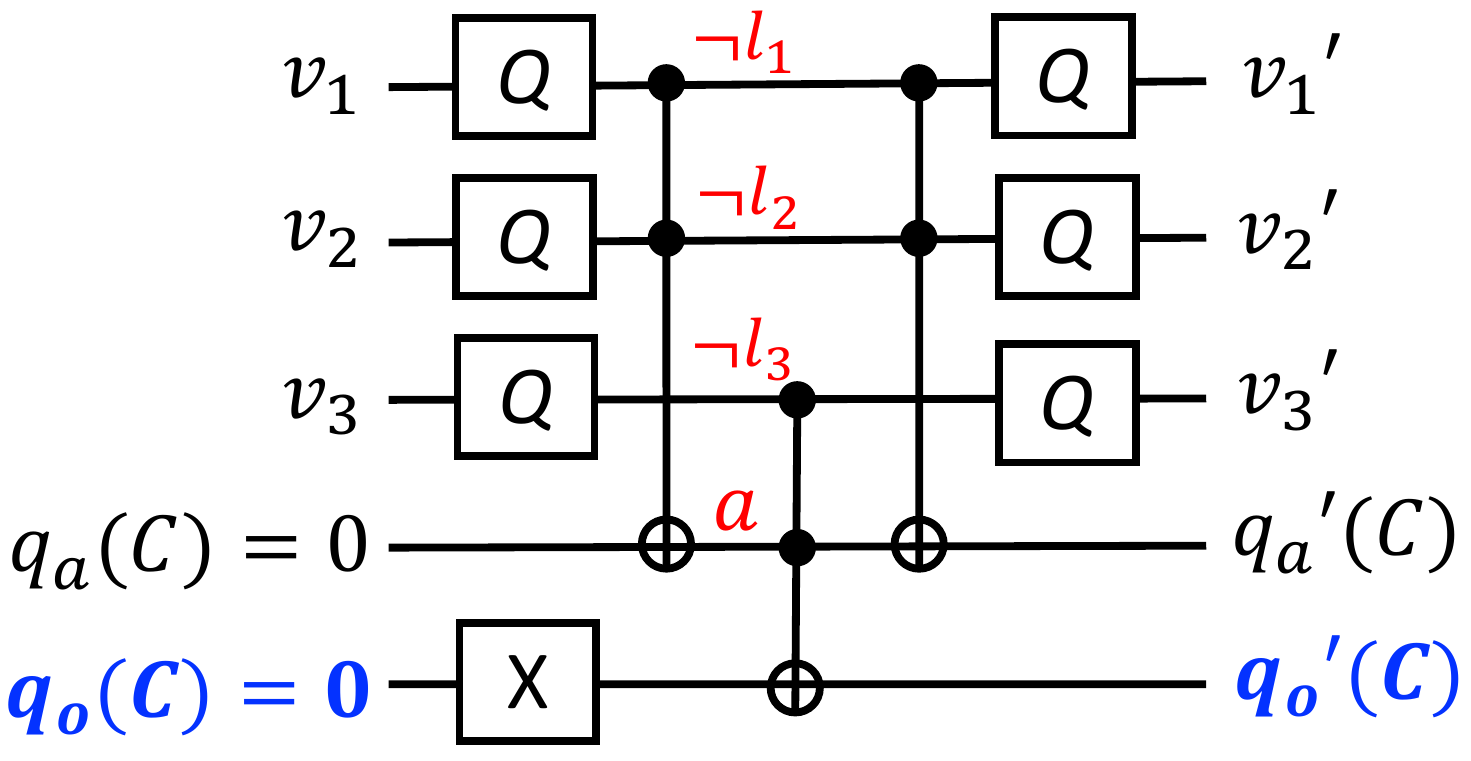}
\end{minipage}
~
\begin{minipage}{0.4\linewidth}
\small
\begin{displaymath}
    Q: \left\{
    \begin{array}{ll}
      \mathtt{X} \mbox{ , if} & l_i \mbox{ is } v_i \mbox{ or } 0 \\
      \mathtt{I} \mbox{ , if} & l_i \mbox{ is } \neg v_i
    \end{array}
    \right.
\end{displaymath}
\end{minipage}
\caption{Quantum circuit for a clause $C: l_1 \vee l_2 \vee l_3$}
\label{fig:SAT_clause_circuit}
\end{figure}

Fig.~\ref{fig:SAT_circuit} shows how to construct a quantum circuit for a formula $F: F_1 \wedge F_2$. It is required to construct the quantum circuits for $F_1$ and $F_2$ first, which are then conjuncted by a $\mathtt{CCNOT}$ gate. The qubit $q_o(F)$ is the output bit for the truth value of formula $F$. Theorem~\ref{thm:SAT-Formula} proves that the quantum circuit construction for a formula $F$ is correct.

\begin{theorem} \label{thm:SAT-Formula}
{\bf [Formula Correctness]} (1) $q_o'(F) = 1$ $\iff$ formula $F$ is true, (2) $q_a'(F) = 0$, and (3) $v_i' = v_i$ for all $i \in \{1,2, \ldots, n\}$.

\end{theorem}

\begin{proof}
We prove this theorem by structural induction on the grammar to generate an arbitrary formula $F$.
The basic case, when $F$ is a clause $C$, has been proved in Theorem~\ref{thm:SAT-Clause}.

Induction assumption : Let $F_1$ and $F_2$ be two formulas satisfying the three conditions: (a) $q_o'(F_1) = 1$ $\iff$ formula $F_1$ is true, and $q_o'(F_2) = 1$ $\iff$ formula $F_2$ is true. (b) $q_a'(F_1) = 0$ and $q_a'(F_2) = 0$. (c) $v_i' = v_i$ for $i \in \{1,2,\ldots, n\}$.

Consider the induction step, the circuit for $F: F_1 \wedge F_2$ is constructed based on $F_1$ and $F_2$, as shown in Figure \ref{fig:SAT_circuit}. 

To prove condition~{(1)}: Because of the $\mathtt{CCNOT}$ gate, $q_o' (F) = 1 \Leftrightarrow ((q_o' (F_1) = 1) \wedge (q_o' (F_2) = 1)) \oplus 0$ $ \Leftrightarrow (q_o' (F_1) = 1) \wedge (q_o' (F_2) = 1)$. Based on induction assumption~{(b)}, we can conclude that $q_o' (F) = 1$ $\iff$ $F_1$ is true and $F_2$ is true $\iff$ $F$ is true.

To prove condition~{(2)}: $q_a' (F) = q_a' (F_2) = 0$, according to the induction assumption~{(b)}.

To prove condition~{(3)}: Since $v_i'$ of $F$ is equal to $v_i'$ of $F_2$, based on induction assumption~{(c)}, $v_i' = v_i$ for all $i \in \{1,2,\ldots, n\}$. \qed
\end{proof}

\begin{figure}[tb]
    \centering
    \includegraphics[width=0.65\linewidth]{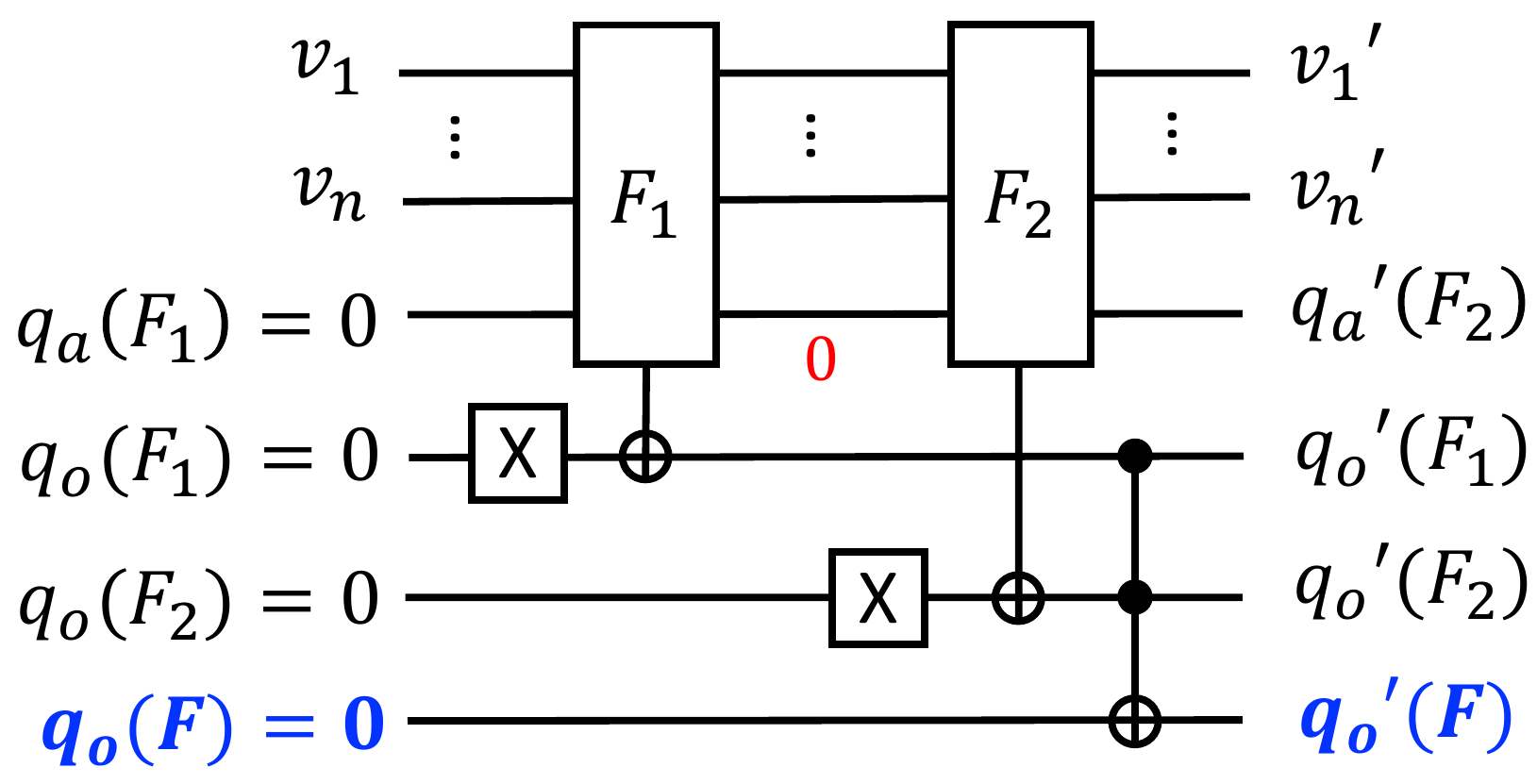}
    \caption{Quantum circuit construction for a formula $F$}
    \label{fig:SAT_circuit}
\end{figure}

\subsection{Theory Circuit} \label{subsec:Theory-Circuit}
Now, we illustrate how to construct the theory circuit for atoms consisting of arithmetic and comparison operations. Consider the syntax grammar for an arbitrary atom, as shown in Fig.~\ref{fig:BV-Syntax}. An expression could be a variable $\Tilde{a}$, a constant $\mathcal{C}$, or the result of an arithmetic operation between two expressions $E_1$ and $E_2$. Constructing the circuit for a variable $\Tilde{a}$ or a constant $\mathcal{C}$ is trivial, which is omitted here. The rest is the case of $\Tilde{a} \odot \Tilde{b}$. Here, we do not list all the arithmetic operations, as there are many. They can either be constructed based on primitive quantum gates or were developed in related works (c.f.~Section~\ref{sec:RelatedWorks}). Instead, we abstract the circuit for arithmetic operations as the general one, shown in Fig.~\ref{fig:arithmetic-comparator}~{(a)}, for easy illustration. Depending on different arithmetic operations, one can obtain its corresponding final circuit constructively in a bottom-up manner.

For each atom of the form $E_1 \rhd E_2$, we adopt the comparator~\cite{OR07} to compare $E_1$ and $E_2$. The overall structure of the comparator is shown in Fig.~\ref{fig:arithmetic-comparator}~{(b)}. Of course, the output $E_1 \odot E_2$ of the arithmetic circuit would be the input of the comparator circuit. The two output bits $O_1$ and $O_2$ of the comparator indicate the relation between $E_1$ and $E_2$, as follows. Notice that the case of $(1,1)$ does not exist in the comparator design~\cite{OR07}.

\begin{displaymath}
(O_1, O_2) = \left\{ 
\begin{matrix}
(0,1) & \mbox{if} & E_1 < E_2 \\
(1,0) & \mbox{if} & E_1 > E_2 \\
(0,0) & \mbox{if} & E_1 = E_2
\end{matrix}
\right.
\end{displaymath}

Based on the two output bits $O_1$ and $O_2$, one can construct the corresponding atom for each of the six different cases, as shown in Fig.~\ref{fig:atom_circuit}~{(a)}--{(f)}. Theorem~\ref{thm:Atom-correctness} shows the correctness of our circuit construction for atoms.

\begin{theorem} \label{thm:Atom-correctness}
{\bf [Atom correctness]} For each atom $E_1 \rhd E_2$ and its corresponding $atom$ bit, we have $atom = 1 \iff (E_1 \rhd E_2) = 1$.
\end{theorem}

\begin{proof}
The proof is straightforward by examining the $atom$ output for each case based on the truth table of primitive quantum gates. We omit it here due to the page limit. \qed
\end{proof}


\begin{figure}[tb]
\begin{minipage}{0.47\linewidth}
\centering
\includegraphics[width=0.6\linewidth]{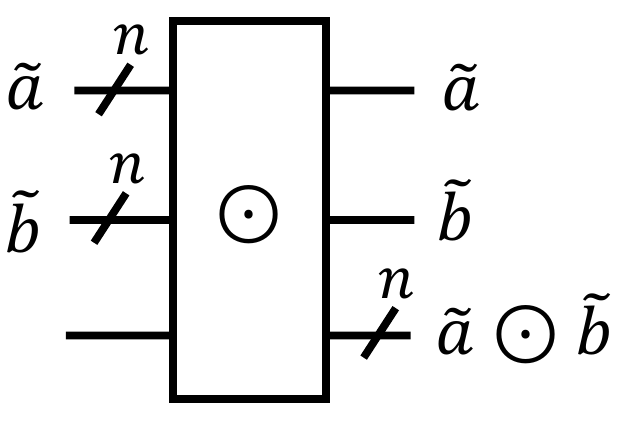} 
\end{minipage}
~
\begin{minipage}{0.47\linewidth}
\centering
\includegraphics[width=0.55\linewidth]{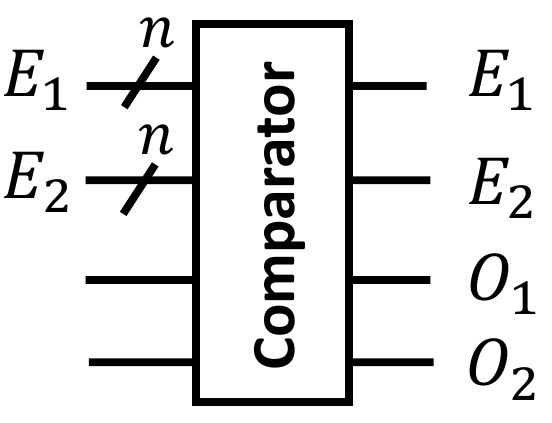} 
\end{minipage}
\\[2mm]
\begin{minipage}{0.47\linewidth}
\centering
\footnotesize
(a) Arithmetic Circuit
\end{minipage}
~
\begin{minipage}{0.47\linewidth}
\centering
\footnotesize
(b) Comparator Circuit
\end{minipage}
\caption{Theory Circuit}
\label{fig:arithmetic-comparator}
\end{figure}

\begin{figure}[tb]
\centering
\includegraphics[width=0.96\linewidth]{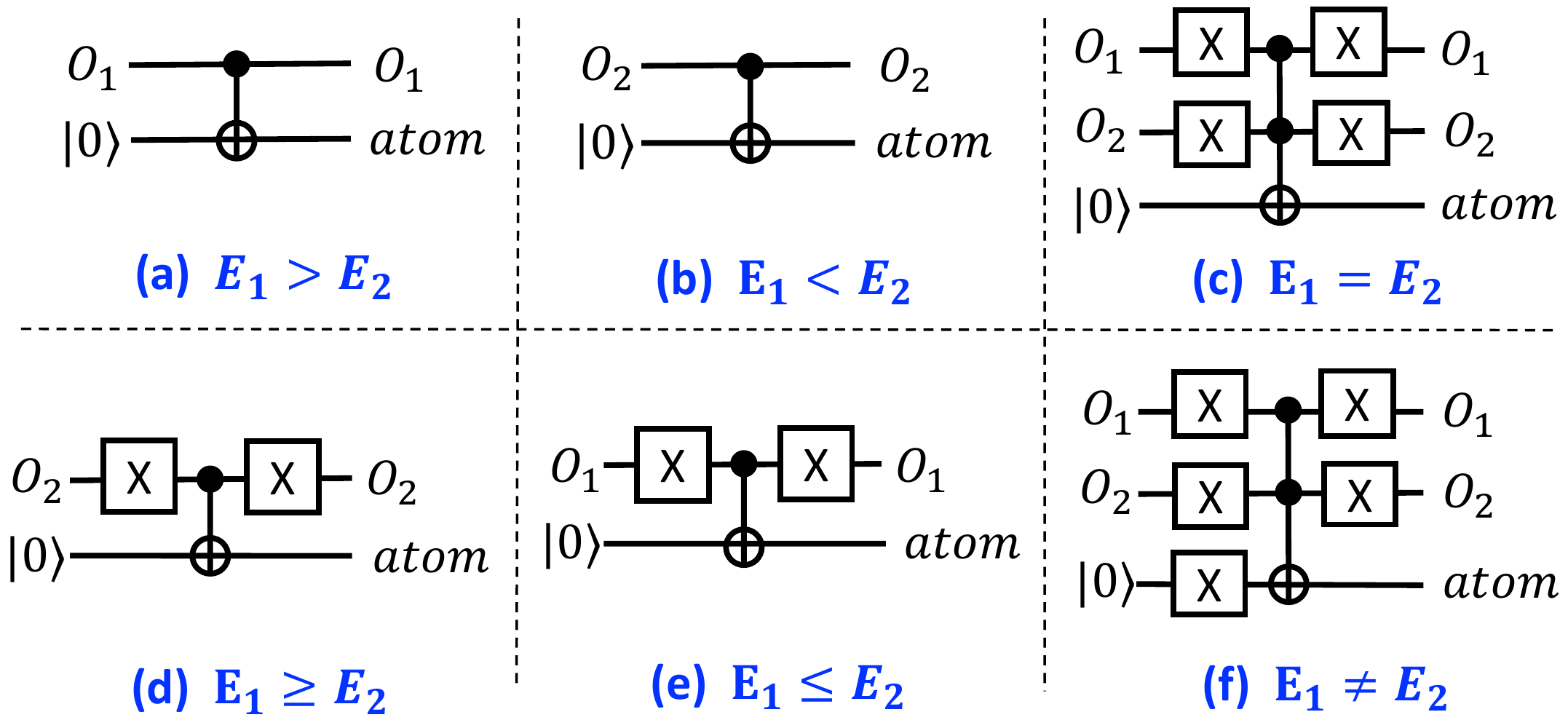}
\caption{Quantum circuit construction for atoms}
\label{fig:atom_circuit}
\end{figure}

\subsection{Consistency Extractor} \label{subsec:Consistency_Extractor}

The task of {\em consistency extractor} is to extract those assignments that are consistent in both Boolean and the bit-vector domains. That is, given an $atom_i$ with its Boolean abstract variable $v_{B_i}$, we want to make sure that $v_{B_i}' = 1$ iff $atom_i$ and $v_{B_i}$ are logical equivalent.


As shown in Fig.~\ref{fig:SMT_extract_circuit}~{(a)}, consistency extractor is composed of a $\mathtt{CNOT}$ gate and a $\mathtt{X}$ gate. The input $atom_i$ serves as the control bit of the $\mathtt{CNOT}$ gate, while the other input $v_{B_i}$ serves as the target bit, followed by a $\mathtt{X}$ gate to flip its result. With this design, the output qubit $v_{B_i}'$ would be $1$ iff $atom_i$ and $v_{B_i}$ have the same truth value. Theorem~\ref{thm:Consistency_Extractor_Correctness} proves the correctness of our quantum circuit for consistency extractor.


\begin{theorem} \label{thm:Consistency_Extractor_Correctness}
{\bf [Correctness of Consistency Extractor]} \\ 
($v_{Bi}' = 1) \iff (v_{B_i} \equiv \; atom_i)$.
\end{theorem}

\begin{proof}
Let $a$ be the result of $v_{Bi}$ after applying the $\mathtt{CNOT}$ gate, as marked in red in Fig.~\ref{fig:SMT_extract_circuit}.
\begin{displaymath}
\begin{array}{lcl}
v'_{Bi} = 1 & \Leftrightarrow & a = 0 \;\; \Leftrightarrow \;\; (v_{Bi} \oplus atom_i) = 0 \\
            & \Leftrightarrow & (v_{Bi} = 0 \wedge atom_i = 0) \vee (v_{Bi} = 1 \wedge atom_i = 1) \\
            & \Leftrightarrow & v_{B_i} \; \equiv \; atom_i \qed
\end{array}
\end{displaymath}
\end{proof}

\begin{figure}[tb]
\begin{minipage}{0.4\linewidth}
\centering
\includegraphics[width=0.95\linewidth]{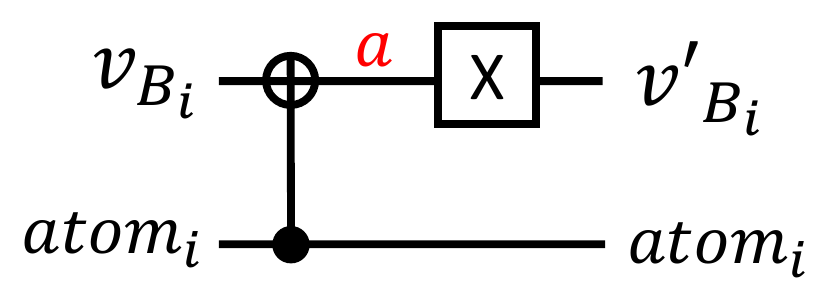}
\end{minipage}
~
\begin{minipage}{0.56\linewidth}
\centering
\includegraphics[width=0.98\linewidth]{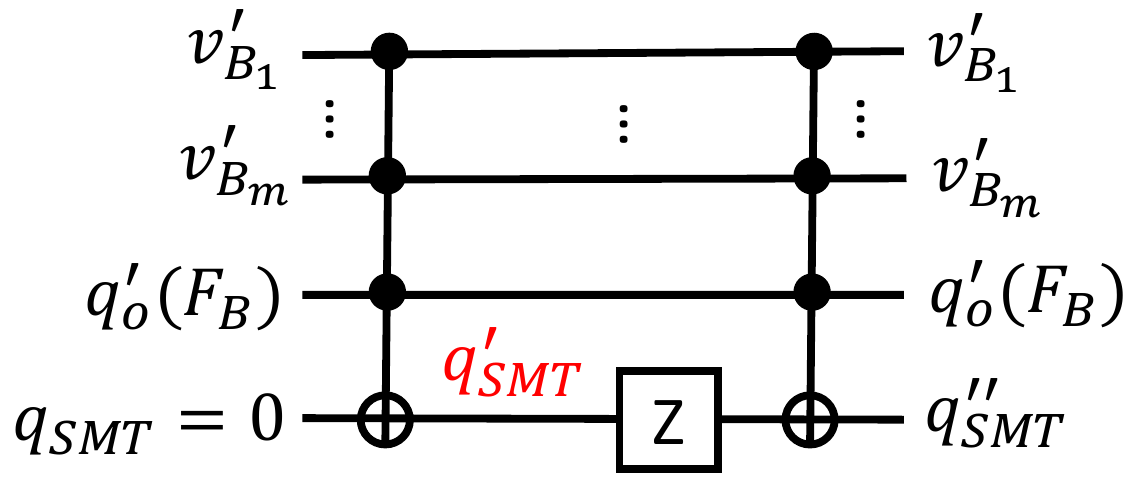}
\end{minipage} \\[2mm]
\begin{minipage}{0.4\linewidth}
\centering
\footnotesize
(a) Consistency Extractor
\end{minipage}
~
\begin{minipage}{0.56\linewidth}
\centering
\footnotesize
(b) Solution Inverter
\end{minipage}
\caption{Consistency Extractor and Solution Inverter}
\label{fig:SMT_extract_circuit}
\end{figure}

\subsection{Solution Inverter} \label{subsec:solution_inverter}

{\em Solution inverter} is the key component of the oracle. It marks the solutions to the SMT formula by inversing them, i.e., giving them a ``$-1$'' phase such that the diffuser knows which elements to increase/maximize their probability being measured.

The circuit for solution inverter, as shown in Fig.~\ref{fig:SMT_extract_circuit}~{(b)}, is composed of two $(m+1)$-$\mathtt{CNOT}$ gates and a $\mathtt{Z}$ gate, where $m$ is the number of Boolean abstract variables. Theorem~\ref{thm:Solution_Inverter_Correctness} proves the correctness of solution inverter.

\begin{theorem} \label{thm:Solution_Inverter_Correctness}
{\bf [Correctness of Solution Inverter]} \\
(1) $q'_{\mathtt{SMT}} = 1$ iff $q'_o(F_B) = 1$ and $v_{B_i} \equiv atom_i$, $\forall i \in \{1,2,\ldots, m\}$. \\
(2) All the solutions are added a ``$-1$'' phase.
\end{theorem}

\begin{proof}
Let $q'_{\mathtt{SMT}}$ be the value of the $q_{\mathtt{SMT}}$ bit after applying the first $(m+1)$-$\mathtt{CNOT}$ gate, as marked in red in Fig.~\ref{fig:SMT_extract_circuit}~{(b)}.
Because $q_{\mathtt{SMT}}$ (initialized as $0$) is the target bit of the $(m+1)$-$\mathtt{CNOT}$ gate, provided that $v_{B_i}'$ is one of the $m$ controlled bit, we know that
$q'_{SMT} = 1 \Leftrightarrow (v'_{B_1} \wedge v'_{B_2} \wedge ... \wedge v'_{B_m} \wedge q_o'(F_B)) = 1$. By Theorem~\ref{thm:Consistency_Extractor_Correctness}, we can conclude that condition~{(1)} holds.

To prove condition~{(2)}, let us consider the state of the quantum circuit. Let $| v'_{B_1}, \ldots, v'_{B_m}, \ldots, q_o'(F_B), q'_{\mathtt{SMT}} \rangle$ be the quantum state after applying the first $(m+1)$-$\mathtt{CNOT}$ gate. As the value of $q'_{\mathtt{SMT}}$ could be $0$ or $1$, the system space $S$ can be split into two disjoint sets $S_0$ and $S_1$. That is, $S = S_0 \cup S_1$, where $S_0 = \{| v'_{B_1}, \ldots, v'_{B_m}, \ldots, q_o'(F_B), 0 \rangle \}$ and $S_1 = \{| v'_{B_1}, \ldots, v'_{B_m}, \ldots, q_o'(F_B), 1 \rangle \}$. By condition~{(1)} we just proved, $S_1$ is exactly the set of all solutions to the SMT problem. After that, a $\mathtt{Z}$ gate is applied on the $q'_{\mathtt{SMT}}$ bit. Since $\mathtt{Z}(|0\rangle) = |0\rangle$ and $\mathtt{Z}(|1\rangle) = -1 |1\rangle$, the $Z$ gate will give a ``$-1$'' phase for each element in $S_1$, i.e., all solutions are added a ``$-1$'' phase. \qed
\end{proof}

\subsection{Reverse circuit}
After inversing the solutions, the last part of the oracle function is the {\em reverse circuit}, which is constructed by the reversed circuits corresponding to the four aforementioned components. 
The purpose of the reverse circuit is to restore the qubits back to their original states, following the reversible nature of quantum computing.
In addition, Grover's algorithm may take several iterations to amplify the probability of solutions. Ancilla bits need to be restored to be used for the following iterations.

\begin{figure*}[htb]
\centering
\begin{minipage}{0.64\linewidth}
\includegraphics[width=\linewidth]{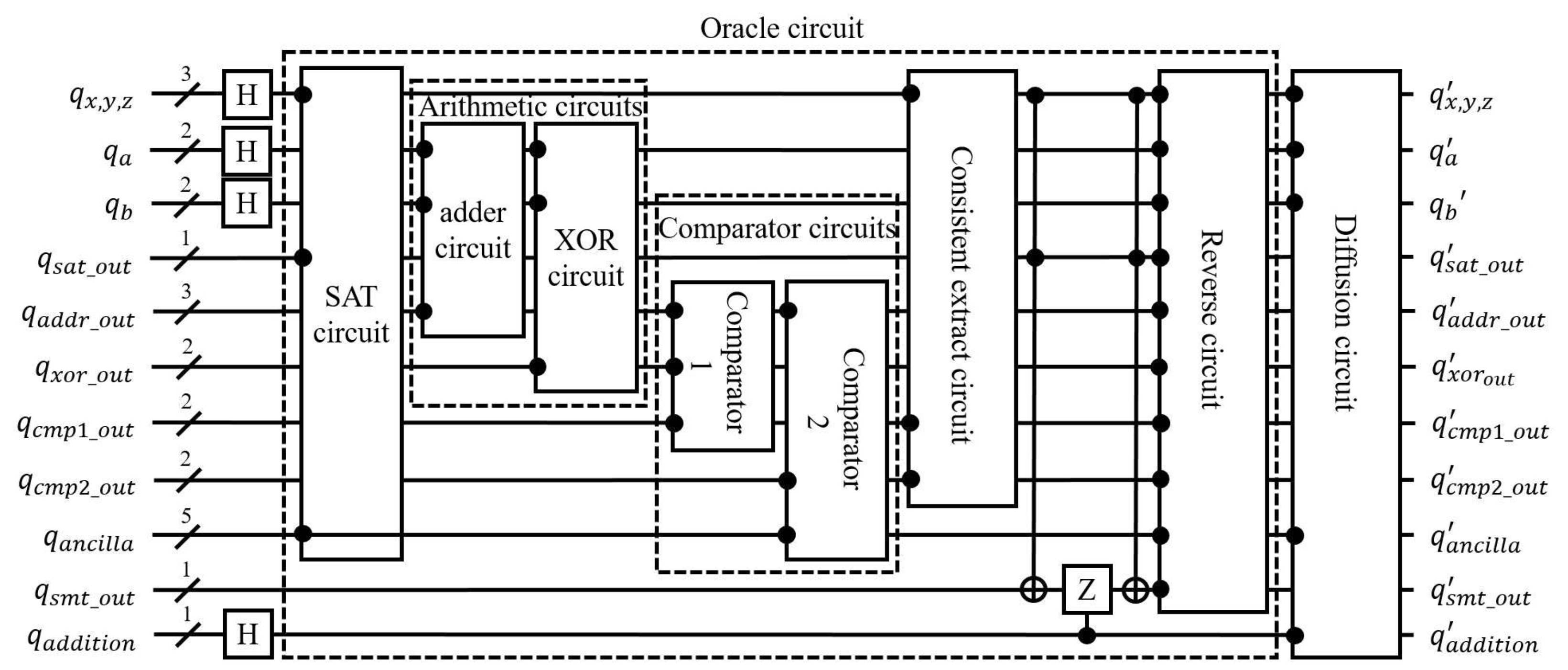}
\end{minipage}
~
\begin{minipage}{0.32\linewidth}
\includegraphics[width=\linewidth]{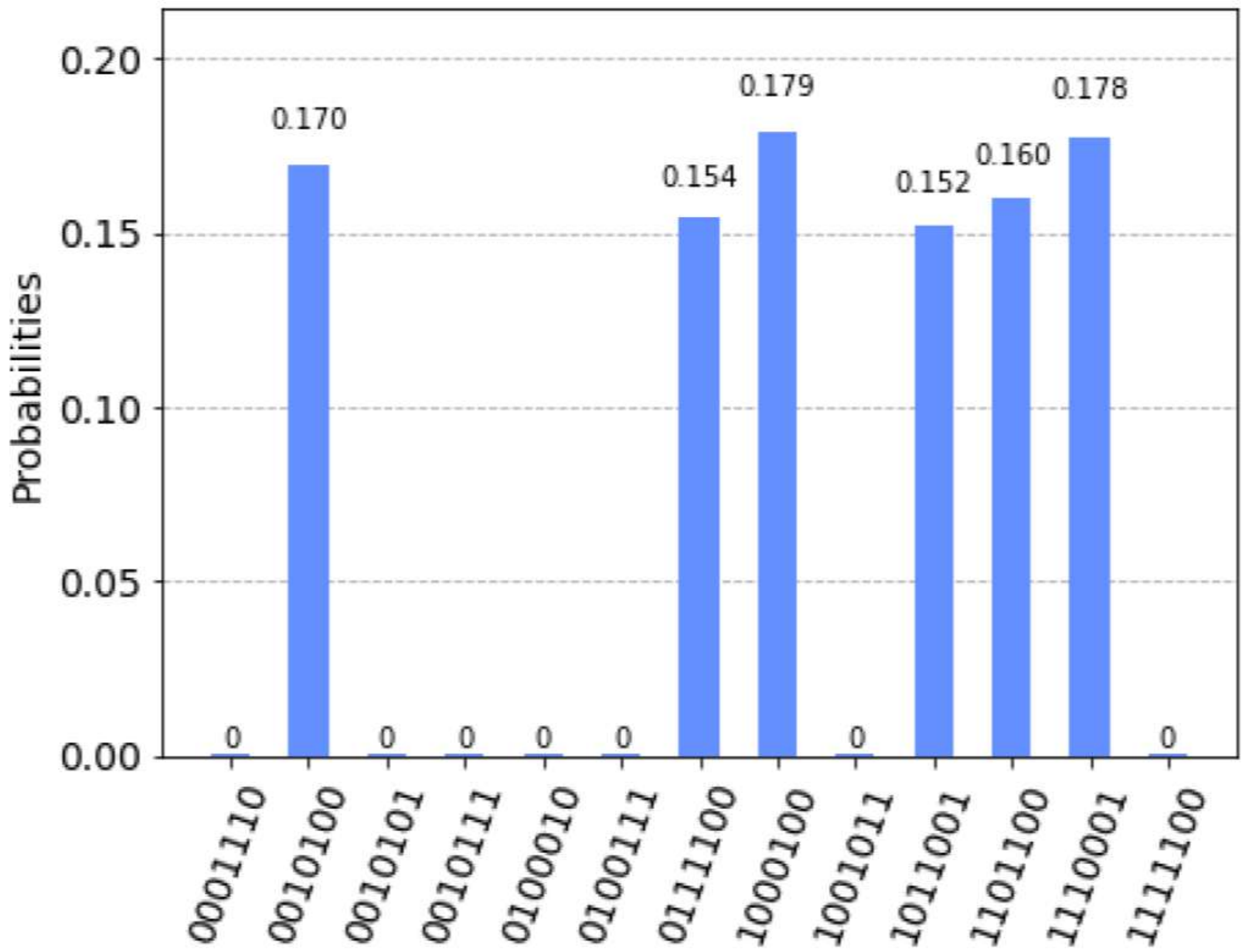}
\end{minipage} \\[1mm]
\begin{minipage}{0.62\linewidth}
\centering
(a)
\end{minipage}
~
\begin{minipage}{0.32\linewidth}
\centering
(b)
\end{minipage}
\caption{The quantum circuit of solving formula $\mathcal{F}$ and the simulation result}
\label{fig:evaluation_circuit}
\end{figure*}

\section{Evaluation} \label{sec:Evaluation}

In this section, we demonstrate our quantum circuit design for solving a $\mathcal{BV}$ SMT formula $\mathcal{F}$, whose Boolean abstract formula is
\[
\mathcal{F}_B: (x \vee y \vee z) \wedge (x \vee \neg y \vee z)
\]
with atom $x: (a + b < a \oplus b)$, $y: (a + b > a \oplus b)$, and $z: (a + b = 1)$, where variables $a$, $b$, $c$ are all $2$-bit long; `$+$' is the modulo-sum operation and `$\oplus$' is the exclusive-or operation.



Fig.~\ref{fig:evaluation_circuit}~{(a)} shows the block diagram of the quantum SMT solver for formula $\mathcal{F}$, in which only important qubits are shown. Internal qubits of each module are omitted. Although the SAT and theory (arithmetic + comparator) circuit are drawn sequentially in the diagram, they can operate in parallel in the realistic implementation since there are no dependency between them.
The main inputs are the three Boolean abstract variables $x, y, z$ and SMT variables $a$, $b$. They are placed in superposition by the Hadamard gate initially. To solve formula $\mathcal{F}$, two comparators are required. Comparator one is responsible for the relation between $(a + b)$ and $(a \oplus b)$, while comparator two is for $(a + b)$ and the constant $1$.

\begin{table}[tb]
\footnotesize
	\begin{center}
		\caption{List of qubits contain in circuit for solving $\mathcal{F}$}
		\label{tb:qubit_list}
		\begin{tabular}{ | c | c | c |}\hline
			Module & Qubit name & Amount \\ \hline
			\multirow{5}{*}{SMT} & Boolean abstract variable & 3 \\ \cline{2-3}
			& SMT variables & 2*2 \\ \cline{2-3}
		    & ancilla qubits & 5 \\ \cline{2-3}
		    & SMT output & 1 \\ \cline{2-3}
		    & addition qubit & 1 \\ \hline
		    \multirow{2}{*}{SAT} & SAT output & 1 \\ \cline{2-3}
		    & extra qubits & 2 \\ \hline
		    Adder & adder output & 3 \\ \hline
		    Bitwise XOR & bitwise XOR output & 2 \\ \hline
		    \multirow{3}{*}{Comparator 1 \& 2} & comparator output & 2*2 \\ \cline{2-3}
		    & comparator internal output & (2-1)*2*2 \\ \cline{2-3}
		    & comparator ancilla & 1*2 \\ \hline
		\end{tabular}
	\end{center}
\end{table}

The circuit was implemented and simulated in Qiskit~\cite{GT19}. The whole circuit requires $32$ qubits, which already reaches the maximum number of qubits that Qiskit supports. The breakdown of the qubits required is listed in Table~\ref{tb:qubit_list}.
The simulation result, as shown in Fig.~\ref{fig:evaluation_circuit}~{(b)}, was obtained by performing five Grover iterations as one shot, repeated for $1,024$ shots to get $1,024$ measurements. 
One may wonder how we get the required number of iterations (i.e., $\frac{\pi}{4}\sqrt{N/M}$, c.f.~Section~\ref{sec:Preliminary}) for performing Grover's algorithm. Actually, this number could be obtained by quantum counting~\cite{BH98}. However, the avaliable $32$ qubits supported by Qiskit are not sufficient to perform quantum counting for this circuit. Thus, we started from one iteration and increased the number until the probability distribution of the measurements started to get worse. The turning point that we got is five, for this experiment. Actually, we have manually calculated the value of $\frac{\pi}{4}\sqrt{N/M}$ by enumerating the search space and the solutions. We found that it is consistent with the number, five.

As shown in Fig.~\ref{fig:evaluation_circuit}~{(b)}, we can observe six bit-strings with a higher probability compared to others. Table~\ref{tb:evaluation_solutions} interprets the six bit-strings back to the assignments of five main variables $x$, $y$, $z$, $a$, $b$. They are exactly the solutions to formula $\mathcal{F}$. 
The “Counts” column represents how many times the assignment was measured among the $1,024$ shots. According to the simulation result, our circuit provides a $99.32\%$ (the summation of ``Counts'' divided by $1024$) of probability of measuring the solutions. In addition, it is possible to obtain all the solutions within a reasonable number of measurement shots. 


\begin{table}[tb]
\scriptsize
	\begin{center}
		\caption{Solutions of $\mathcal{F}$ obtained by quantum circuit}
		\label{tb:evaluation_solutions}
		\begin{tabular}{ | c | c | c | c | c |}\hline
			Output & \multirow{2}{*}{Counts} & \multicolumn{3}{c|}{Assignments} \\ \cline{3-5}
			bit-string & & $(x,y,z)$ & $a$ & $b$ \\ \hline
			$0010100$   &   $174$  & $(0,0,1)$     & $01$  & $00$ \\ \hline
			$0011110$   &   $158$  & $(0,0,1)$	    & $11$ 	& $10$ \\ \hline
			$0010001$   &   $183$  & $(0,0,1)$		& $00$ 	& $01$ \\ \hline
			$1001101$   &   $156$  & $(1,0,0)$		& $11$ 	& $01$ \\ \hline
			$0011011$   &   $164$  & $(0,0,1)$		& $10$ 	& $11$ \\ \hline
			$1000111$   &   $182$  & $(1,0,0)$		& $01$ 	& $11$ \\ \hline
		\end{tabular}
	\end{center}
\end{table}

\section{Related Works} \label{sec:RelatedWorks}

\subsection{SAT circuit} \label{subsec:SAT-Circuit}
SAT is useful in a wide range of applications, but it is also a famous NP-Complete problem. There are many works, such as hardware implementation of SAT solvers, trying to improve the performance of solving SAT problems.
With the advantage of superposition in the quantum system, there exists a chance to speed up the process of SAT solving by using the quantum circuit. Alberto et al. \cite{AS07} proposed three algorithms, including the quantum Fredkin circuit, quantum register machine and quantum P system, to solve the SAT problem. By adopting the superposition nature of the quantum system to the input, the system can obtain all the results of all possible inputs in one evaluation. But the common problem of those methods is that it is difficult to measure the SAT result since the probability of SAT results will be very low if the number of satisfiable assignments is way less than the number of unsatisfiable assignments. Although the author proposed a non-unitary operator or there are other methods like the chaotic dynamic system \cite{MI03} to extract the SAT result, they are difficult to implement for now.

Another approach to solving the measurement problem is to adopt Grover’s search algorithm. Since the algorithm is used to find the target elements in a database by amplifying the measurement probability of those elements, it is suitable for searching the satisfiable assignments if those assignments can be marked properly in the oracle function. Fernandes et al.~\cite{FS19} provide the concept of constructing the oracle circuit of 3-SAT solving based on Grover’s algorithm.

\subsection{Quantum arithmetic and comparator circuit}

Chakrabarti et al.~\cite{CS08} proposed a quantum ripple carry adder that is based on the famous VBE adder~\cite{VB96} with the improvement of gate level reduction. Their proposed adder reorganizes the carry gates that generate carry bits to let some parts of the gates operate in parallel. On the other hand, they discard the series of sum and carry gates in the second part of the VBE adder, and replace those gates with $\mathtt{CNOT}$ gates to generate the summation result. The author also proposed the design of a quantum carry look ahead adder, but there are no benefits of qubit usage or gate level compared to the ripple carry adder. We adopt the ripple carry adder for modulo-sum operation in our circuit.

Kotiyal et al.~\cite{KT14} proposed a design of quantum multiplier focusing on reducing the ancilla inputs and garbage outputs since the usage of qubits is the main consideration of quantum circuit design. The advantage of their proposed multiplier is that they use the binary tree based design to handle partial product summation in parallel. Besides, they employ the adder without ancilla and garbage qubits proposed by Takahashi et al.~\cite{TT10} to reduce the caused by adders. The implementation result also showed that their multiplier has up to $60.34\%$ and $52.27\%$ improvement of ancilla and garbage qubits usage compared to other previous works.

Kotiyal et al. \cite{KT11} proposed a quantum bidirectional barrel shifter that can perform four shift operations, including logical right shift, logical left shift, arithmetic right shift and arithmetic left shift. In order to implement the 2to1 MUX, which is the major component of the barrel shifter, they adopt the Fredkin gate (i.e. $\mathtt{CSWAP}$ gate) for implementation by setting the control qubit as the select signal and one of the target qubits as the output. They also use the $\mathtt{CNOT}$ gate to handle signal fan out. Compared to previous works, their shifter has the reversal control unit to achieve the bidirectional shift even though the core module of the shifter only operates right shift operation.

Finally, the quantum comparator will be used to identify the atom’s status once the two $\mathcal{BV}$ terms are ready. The quantum comparator proposed by Oliveira et al. \cite{OR07} can indicate the relation between two input quantum bit strings $\mathbf{a}$ and $\mathbf{b}$. Their circuit, compared to another subtraction based comparator NKO, requires less resources and has higher computation parallelism.

\section{Conclusion} \label{sec:Conclusion}

In this paper, we proposed a quantum SMT solver for the $\mathcal{BV}$ theory, based on Grover’s algorithm. 
Our proposed approach not only provides a theoretically quadratic speedup compared to the traditional method, but also has the capability of discovering all solutions. 
Since solving SMT problems for other theories also consists of two parts (namely SAT and theory solving), our proposed oracle construction is general to be extended to other theories, which is our future work.

\bibliographystyle{plain}
\bibliography{reference-DAC}

\end{document}